\documentclass{IEEEtran}
\usepackage{amsthm}
\usepackage{amsmath,amssymb,amsfonts,mathtools}
\usepackage{graphicx}

\usepackage[english]{babel}   
\usepackage{tikz}

\usetikzlibrary{shapes, arrows.meta, positioning}

\usepackage{fancyhdr}
\fancypagestyle{firstpage}{
    \fancyhf{}
    
    \fancyfoot[C]{\small This work has been submitted to the IEEE for possible publication. \\Copyright may be transferred without notice, after which this version may no longer be accessible.}
}

\usepackage{xcolor}
\usepackage{hyperref}

\usepackage{siunitx}
\definecolor{myblue}{HTML}{0072BD}
\definecolor{myred}{HTML}{D95319}
\definecolor{myyellow}{HTML}{EDB120}
\definecolor{mypurple}{HTML}{7E2F8E}
\definecolor{mygreen}{HTML}{77AC30}

\newtheorem{theorem}{Theorem}
\newtheorem{proposition}{Proposition}
\newtheorem{definition}{Definition}
\newtheorem{remark}{Remark}
\newtheorem{assumption}{Assumption}
\newtheorem{lemma}{Lemma}

\newtheorem{objective}{Objective}

\usepackage[style=ieee, % numerische Sortierung
citestyle=ieee-comp, % gruppiert aufeinanderfolgende Referenzen: [1]-[4]
backend=bibtex,  % Bibtex-Format
isbn=false,  % Unterdrückung ISBN-Angabe
doi=false, % Unterdrückung DOI-Angabe
giveninits=true,url=false,date=year, maxbibnames=99, minnames=1, maxnames=3, sorting=none,alldates=year,]{biblatex} 
\AtEveryBibitem{\clearfield{pages}}\AtEveryBibitem{\clearfield{volume}}

\def\BibTeX{{\rm B\kern-.05em{\sc i\kern-.025em b}\kern-.08em
    T\kern-.1667em\lower.7ex\hbox{E}\kern-.125emX}}
\begin{document}
\title{Dynamic State-Feedback Control for LPV Systems: Ensuring Stability and LQR Performance}
\author{A. Gießler, \IEEEmembership{Member, IEEE}, F. Strehle, J. Illerhaus, and S. Hohmann, \IEEEmembership{Senior Member, IEEE}
%\thanks{Placeholder: submission date }
\thanks{A. Gießler, J. Illerhaus, and S. Hohmann are with the Institute of Control Systems, Karlsruhe Institute of Technology, 76131 Karlsruhe, Germany (e-mail: \{armin.giessler, jochen.illerhaus, soeren.hohmann\}@kit.edu).}
\thanks{F. Strehle was with the Institute of Control Systems, Karlsruhe Institute of Technology, 76131 Karlsruhe, Germany. He is now with Netze BW, 70567 Stuttgart, Germany  (e-mail: f.strehle@netze-bw.de).}}

\maketitle
\thispagestyle{firstpage}
\begin{abstract}
In this paper, we propose a novel dynamic state-feedback controller for polytopic linear parameter-varying (LPV) systems with a constant input matrix. 
The controller employs a projected gradient flow method to continuously improve its control law and, under established conditions, converges to the optimal feedback gain of the corresponding linear quadratic regulator (LQR) problem associated with constant parameter trajectories.
We derive conditions for quadratic stability, which can be verified via convex optimization, to ensure exponential stability of the LPV system even under arbitrarily fast parameter variations. 
Additionally, we provide sufficient conditions to guarantee the boundedness of the trajectories of the dynamic controller for any parameter trajectory and the convergence of its feedback gains to the optimal LQR gains for constant parameter trajectories.  
Furthermore, we show that {incorporating the plant state into the controller dynamics renders} the closed-loop system asymptotically stable for constant parameter trajectories under these conditions.
Simulation results demonstrate that the controller maintains stability and improves transient performance.
\end{abstract}


\section{Introduction}
\label{sec:introduction}

Systems with varying dynamics due to parameter variations are prevalent in control engineering and find numerous applications in domains such as aerospace, automotive, robotics, and power systems \cite{Hoffmann}.
% Linear Parameter-Varying (LPV) systems have emerged as a powerful framework for modeling and controlling certain classes of nonlinear dynamical systems subject to parameter variations. This framework employs linear state-space models whose matrices are functions of time-varying parameters \cite{mohammadpour2012control,briat_linear_2015}. 
% Traditional LPV systems incorporate exogenous, typically measurable parameters, whereas quasi-LPV systems feature parameters that depend on states, inputs, or outputs \cite{Leith01012000}. This flexibility allows \mbox{(quasi-)} LPV systems to capture a wide range of nonlinear behaviors while retaining the mathematical tractability and applicability of linear system methods. 
Linear parameter-varying (LPV) systems have emerged as a powerful framework for modeling and controlling certain classes of nonlinear dynamical systems subject to parameter variations.
This framework employs linear state-space models whose matrices are functions of time-varying, typically measurable exogenous parameters \cites{mohammadpour2012control,briat_linear_2015}. This formulation enables LPV systems to capture a broad range of nonlinear behaviors while retaining the mathematical tractability and applicability of linear system methods.

It is widely recognized that parameter variations can induce instability in an LPV system, even if all frozen-time systems\footnote{Frozen-time systems of an LPV system refer to the family of linear time-invariant (LTI) systems obtained by fixing the parameter trajectory at a constant value within the defined parameter set \cite{shamma1988phd,briat_linear_2015}.} are asymptotically stable%\footnote{Strictly speaking, the asymptotic stability of the LPV system refers to the stability of its equilibrium point at the origin.} 
 \cite{Shamma1992}. In the context of LPV systems, a distinction is made between slowly-varying and (arbitrarily) fast-varying parameters \cite[Sec. 1.2.2]{mohammadpour2012control}. 
Quadratic stability is of particular importance, as it ensures exponential stability of an LPV system for both slow and fast parameter variations by finding a parameter-independent Lyapunov function valid across all admissible parameter trajectories \cite[Sec. 2.3.1]{briat_linear_2015}\cite[Sec. 1.2]{wu1995phd}.
Further stability conditions for systems with slowly-varying parameters are presented in \cite{ILCHMANN1987} and \cite[Sec. 1.2.2]{mohammadpour2012control}.

\subsection{Literature Review}
Drawing on quadratic stability, extensive research has been conducted on stabilizing control methods for LPV systems \cite[Sec. 3.1]{mohammadpour2012control}\cite[Sec. 1.1]{white_linear_2013}.
The classical approach is gain scheduling,
where the feedback gain adapts to varying parameters \cite{Leith01012000,rugh_research_2000}. 
% Most methods use a common Lyapunov function valid for all parameter trajectories to ensure quadratic stability \cite{Leith01012000}, though some employ parameter-dependent \cite{wu96Induced,apkarian_advanced_1998,Feron1996} or polyhedral functions \cite{Blanchini2003}.
Most methods exploit a common Lyapunov function for quadratic stability \cite{Leith01012000}, with some exceptions that employ parameter-dependent \cite{wu96Induced,apkarian_advanced_1998,Feron1996} %YU199757
or polyhedral Lyapunov functions \cite{Blanchini2003}.
Gain-scheduled controllers are broadly categorized into state-feedback \cite{SHAHRUZ1992,wu1995phd}, output-feedback \cite{scherer1996mixed} and dynamic output-feedback approaches \cite{BECKER1994,wu1995phd,apkarian_advanced_1998,Apkarian1995,APKARIAN19951251,wu96Induced,Courties1999}.

Many dynamic output feedback controllers guarantee a bounded induced $L_2$-norm performance from disturbance to error signals, utilizing robust $H_\infty$ control techniques for controller synthesis \cite{APKARIAN19951251,Apkarian1995,wu96Induced,apkarian_advanced_1998}.
However, while robustness to disturbances is a common focus, relatively few methods address $H_2$ performance objectives \cite{Courties1999,scherer1996mixed}. 
The mixed $H_2 /H_\infty$ output feedback control in \cite{scherer1996mixed} represents a compromise between $H_2$ and $H_\infty$ performance and hence does not achieve optimal $H_2$ objective. 
The $H_2$ dynamic output feedback control in \cite{Courties1999} is typically not optimal because blending the optimal controllers at the vertices of the parameter set does not necessarily yield the optimal performance.
%{For switched linear systems, a static state feedback based on the Riccati equation of a nominal mode stabilizes the system under arbitrary switching \cite{Bonilla2021}, but is LQR-optimal only for the nominal mode.}
To the best of our knowledge, no existing controller for LPV systems simultaneously minimizes the LQR objective\footnote{Note that the LQR objective is a special case of $H_2$ minimization under state feedback \cite{Feron}.} while guaranteeing closed-loop stability. %\cite{Feron} 
\subsection{Contributions}
The contributions of this paper are fourfold:
\begin{enumerate}
    \item We propose a novel dynamic state-feedback controller tailored to polytopic LPV systems with a constant input matrix. 
    The controller employs a projected gradient flow, continuously improves its control law and, under established conditions, converges to the optimal LQR feedback gain for  constant parameter trajectories.
   \item  \label{item:second}
   Under the assumption of bounded trajectories of the dynamic controller, we derive sufficient and necessary conditions for the quadratic stability of the LPV system that are verifiable via convex optimization.
   \item \label{item:third} We establish sufficient conditions to ensure the boundedness of the trajectories of the dynamic controller for arbitrarily fast parameter variations. Additional conditions are provided to ensure the convergence of the feedback gains to the optimal LQR gains under constant parameter trajectories.
   \item \label{item:fourth}  {We show that adapting the controller dynamics by augmenting the underlying cost with the LQR cost evaluated at the current plant state renders the frozen-time closed-loop system asymptotically stable under the given conditions.}
 \end{enumerate}
 
\subsection{Paper Organization} 
This section continues with notation. Preliminaries are presented in Section \ref{sec:pre}. Section \ref{sec:2} defines the dynamic state-feedback controller and the closed-loop system, and states the main control objective. In Section \ref{sec:stability}, the stability of the closed-loop system and the optimality of the controller are analyzed. Section \ref{sec:4} presents simulation results. Finally, Section \ref{sec:conc} ends with a conclusion.

\subsection{Notation}
The set of (nonnegative) real numbers is denoted by $\mathbb{R}$ $(\mathbb{R}_{\geq 0})$.
For a matrix $A$, $A^\top$, $\operatorname{rank}(A)$, $\operatorname{tr}(A)$, $\operatorname{vec}(A)$, and $\operatorname{det}(A)$ denote its transpose, rank, trace, vectorization, and determinant, respectively. 
A symmetric positive definite (semidefinite) matrix $A$ is denoted by $A\succ 0$ $(A\succeq 0)$. The gradient of a scalar function $f$ with respect to a vector $x$ (matrix $A$) is denoted by $\nabla f(x)$ ($\nabla f(A)$), and the Jacobian matrix of a vector function $g$ is given by $J_g$. 
The operator $ \operatorname{diag}(\cdot)$ constructs a diagonal matrix from a vector.
The interior, boundary, and cardinality of a set $\mathcal{S}$ are denoted by $\operatorname{int}(\mathcal{S})$, $\partial \mathcal{S}$, and $\vert S \vert$, respectively. For $\mathcal{S}\subset \mathbb{R}^{n\times m}$,  $\operatorname{vec}\left(\mathcal{S} \right) \! \coloneq \! \{\operatorname{vec}\left(X \right)\in\mathbb{R}^{mn} \! \mid \! X\! \in \! \mathcal{S}\}$ denotes its vectorized representation.

% The interior and boundary of a set $\mathcal{S}$ are denoted by $\operatorname{int}(\mathcal{S})$ and $\partial \mathcal{S}$, respectively. 
% The cardinality of a set $S$ is denoted by~$\vert S \vert$.

\section{Preliminaries}
\label{sec:pre}
\subsection{LPV Systems}
Given a compact parameter set $\mathcal{P}\subset \mathbb{R}^p$, the parameter variation set $\mathcal{F}_{\mathcal{P}}$ denotes the set of all admissible parameter trajectories $\rho:\mathbb{R}_{\geq 0}\to \mathcal{P}$. The set $\mathcal{F}_{\mathcal{P}}^\infty$ contains all piecewise continuous trajectories $\rho(t)$ with a finite number of discontinuities, representing arbitrarily fast-varying parameters.
For slowly-varying parameters, $\mathcal{F}_{\mathcal{P}}^\mathcal{V}$ contains continuous trajectories $\rho(t)$, where $\dot \rho(t)$ is an element of a convex and compact polytope $\mathcal{V}$ containing 0. The set $\mathcal{F}_{\mathcal{P}}^{\text{pc}}$ contains all piecewise constant trajectories $\rho(t)$. Hence, $\mathcal{F}_{\mathcal{P}}^{\text{pc}}\subset \mathcal{F}_{\mathcal{P}}^\infty$ and $\mathcal{F}_{\mathcal{P}}^\mathcal{V}\subset \mathcal{F}_{\mathcal{P}}^\infty$.

A generic autonomous LPV system is given by 
\begin{align}
   \dot x(t) = A(\rho(t))x(t), \quad x(0)=x_0,\label{math:LPV} 
\end{align}
where $\rho(t)\in\mathcal{F}_{\mathcal{P}}^\infty$ %\footnote{For this family of parameter trajectories, solutions of \eqref{math:LPV} exists for all time in the Carathéodory sense \cite[Th. 1.1 of Ch. 2]{coddington_theory_1987}.} 
  and $A: \mathbb{R}^p \to \mathbb{R}^{n\times n}$ is a continuous function. The system \eqref{math:LPV}, more precisely the equilibrium at the origin, is called asymptotically frozen-time stable 
  %if the system matrix $A(\rho)$ has eigenvalues with strictly negative real parts for all $\rho\in\mathcal{P}$.
    if the system matrix $A(\rho)$ is a Hurwitz matrix for all $\rho\in\mathcal{P}$.

% \cite[Th. 1.1 of Ch. 2]{coddington_theory_1987}.} 
\begin{proposition}[\texorpdfstring{\cite[Prop. 1.1]{mohammadpour2012control}}{}]
   \label{prop:quad}
   If there exists a matrix $X\succ 0$ such that 
   $ A(\rho)^\top X + X A(\rho) \prec 0$ for all $ \rho \in \mathcal{P},$
   % \begin{align}
   %    A(\rho)^\top X + X A(\rho) \prec 0 \quad \forall \rho \in \mathcal{P},
   % \end{align} 
   then the system \eqref{math:LPV} is exponentially stable for all $\rho(t)\in \mathcal{F}_{\mathcal{P}}$ and is called quadratically stable.
\end{proposition}
% \begin{proof}
%    The system \eqref{math:LPV} has the common Lyapunov function $V(x) = x^\top X x$.
% \end{proof}

If the function $A(\rho(t))$ is affine in $\rho(t)$, i.e., 
\begin{align}
   A(\rho(t)) = A_0 + \sum_{i=1}^{p} \rho_i(t) A_i, \label{math:affine}
\end{align}
and $\mathcal{P}$ is a convex and compact polytope, the system \eqref{math:LPV} is called polytopic LPV system. 
By using the standard $N$-unit simplex, defined as $\Lambda_N \!= \!\{x\in\mathbb{R}^N_{\geq 0}\!: \! \sum_i x_i =1\}$, polytopic LPV systems can be represented by a time-varying convex combination of linear time-invariant (LTI) systems
\begin{align}
   \label{math:LTI_comb}
   A(\rho(t))=\sum_{i=1}^{\vert V \vert} \lambda_i(t) \tilde{A}_i, \quad \lambda(t)\in\Lambda_{\vert V \vert},
\end{align}
where $V= \operatorname{vert }(\mathcal{P})$ denotes the vertices of $\mathcal{P}$.

\begin{theorem}[\texorpdfstring{\cite[Thm. 2.5.1]{briat_linear_2015}}{}]
   \label{theo:poly_quad}
   The polytopic LPV system $\dot x = A(\rho(t)) x$ is quadratically stable if and only if there exists an $X\succ 0$ such that $\tilde{A}^\top_i X + X \tilde{A}_i \prec 0 $ holds for all $i=1,\dots, \vert V \vert$.
\end{theorem}

\subsection{Policy Gradient Flow for Linear Quadratic Regulator}
\label{subsubsec:policy_gradient}
We briefly recall the LQR problem for LTI systems
   \begin{subequations}
            \label{math:LQR}
   \begin{align}
      \min_{x,u}~ & \int_{0}^{\infty} u^\top(t) R u(t) + x^\top(t) Q x(t) \, \mathrm{d}t \\
      \text{s.t.}~ & \dot x(t) = A x(t)+B u(t), \\
      & x(0) = x_0,
   \end{align}
\end{subequations}
where $Q\succeq 0$ and $R\succ 0$ are weighting matrices.
If the system $(A,B)$ is stabilizable and the pair $(A,\sqrt{Q})$ is detectable%\footnote{The Popov-Belevitch-Hautus test for stabilizability and detectability of LTI systems can be found in \cite[Th. 14.3 and Th. 16.6]{hespanha2018}.}
, the optimal solution $u^*{(t)} =-K^* x{(t)}$ stabilizes $(A,B)$. The optimal feedback is $K^*=R^{-1}B^\top P^*$, where $P^*$ is the unique  positive semidefinite solution of the continuous-time algebraic Riccati equation (CARE)
\begin{align}
   \label{math:CARE}
   A^\top P + PA - PB R^{-1} B^\top P +Q = 0. 
\end{align}
%The optimal value function is $V^*(x_0)=x_0^\top P^* x_0$.  
%If the stabilizability and detectability conditions hold, the CARE has a unique positive definite solution. 
%
% The equation \eqref{math:CARE} can be formulated as a Lyapunov equation by substituting $K = R^{-1}B^\top P$
% \begin{align}
%    A_K^\top P + P A_K + Q + K^\top R K  &=0,  \label{math:Ly}
% \end{align}
% where  $A_K := A-BK$.
By substituting $K = R^{-1}B^\top P$, \eqref{math:CARE} can be rewritten as 
\begin{align}
   A_K^\top P + P A_K + Q + K^\top R K  &=0,  \label{math:Ly}
\end{align}
where  $A_K \coloneq A-BK$. Note that \eqref{math:Ly} has the form of a Lyapunov equation if $K$ is considered fixed, i.e., independent of $P$.
The problem \eqref{math:LQR} can be solved using the policy gradient flow from \cite{bu2020clqr}, in which the LQR cost is parametrized in $K$
\begin{align}
   \label{math:LQR_cost}
   f_K(x_0) = x_0^\top P_K x_0 = \operatorname{tr}(P_K x_0 x_0^\top),
\end{align}
where $P_K$ %\footnote{The subscript $K$ of matrices emphasize that the matrix depends on feedback gain $K$. An exception is the definition $A_K = A-BK$.}
is the solution of \eqref{math:Ly} for a (non-optimal) feedback $K$. Since the optimal feedback $K^*$ is independent of the initial state, we replace $x_0 x_0^\top $ by $ I_n$ (see \cite[Sec. 3.3]{bu2020clqr} for details).
The set of stabilizing feedback gains of the system $(A,B)$ is path-connected, unbounded \cite[Section 3]{bu2019} and is denoted by 
\begin{align}
   \mathcal{K} & = \{ K \in \mathbb{R}^{m \times n} : \operatorname{Re}(\lambda_i(A-BK)) < 0 ~ \forall i \}. \label{math:K}
\end{align}
%and is in general regular open, unbounded and path-connected \cite[Section 3]{bu2019}.
The LQR cost $f_K \!= \!\operatorname{tr}(P_K)$ is minimized by the gradient flow
\begin{align}
      \dot K {(t)} = - \nabla f_K{(t)}, \quad  K(0)  = K_0\in\mathcal{K},  \label{math:grad}
\end{align}
where $\nabla f_K  = 2 \left( R K - B^\top P_K \right) Y_K$ and $ Y_K  = \int_0^\infty e^{A_K t} I_n e^{A_K^\top t} \, \mathrm{d}t$ \cite[Lem.~4.1]{bu2020clqr}.
% \begin{align*}
%    \nabla f_K & = 2 \left( R K - B^\top P_K \right) Y_K, \\
%    Y_K & = \int_0^\infty e^{A_K t} x_0 x_0^\top e^{A_K^\top t}\mathrm{d}t.
% \end{align*}
The matrix $Y_K$ can alternatively be obtained by solving the Lyapunov equation 
\begin{align}
   A_K Y_K + Y_K A_K^\top + I_n = 0. \label{math:Ly2}
\end{align}
%By vectorizing the Lyapunov equations \eqref{math:Ly} and \eqref{math:Ly2}, 
%The gradient flow \eqref{math:grad} can be expressed as a closed-form ODE, where the right-hand side is a rational polynomial in the entries of $K$, and is well-defined for $K\in\mathcal{K}$.
% The gradient flow \eqref{math:grad} is a closed-form ODE with a rational function  on the right-hand side and is well-defined for $K\in\mathcal{K}$.
The flow \eqref{math:grad} is a ordinary differential equation (ODE) with a rational function  on the right-hand side, is well-defined for $K\in\mathcal{K}$, and generates trajectories $K(t)$ that remain in $\mathcal{K}$.
%The trajectory $K(t)$ of \eqref{math:grad} converges to the optimal feedback $K^*$ if the initial feedback $K_0$ stabilizes the system $(A,B)$, i.e., $K_0\in\mathcal{K}$.

\subsection{Projected Gradient Flow}
\label{subsec:proj}
Consider the inequality-constrained optimization problem 
\begin{align}
   \min_{ {\theta }\in\mathbb{R}^n}  ~f( {\theta }) \quad  \text{s.t.} \quad  g({\theta }) \geq  0, \label{math:opt1}
\end{align}
where $f:\mathbb{R}^n\to \mathbb{R}$ and $g:\mathbb{R}^n\to \mathbb{R}^l$ are smooth functions, i.e., $f,g \in C^\infty(\mathbb{R}^n)$. 
The feasible region is defined by the  manifold $\mathcal{M}=\{{\theta }\in\mathbb{R}^n \mid g({\theta }) \geq 0 \}$ which is assumed to be compact and connected. Additionally, it is assumed that the Linear Independence Constraint Qualification (LICQ) is satisfied at all ${\theta }\in \mathcal{M}$ 
and that all Karush-Kuhn-Tucker (KKT) points are nondegenerate critical points (see \cite[Sec. 2]{jongen2003constrained}).
{Let
\begin{align}
   \Pi_{C_\theta \mathcal{M}}(-\nabla f(\theta))=  \arg \min_{\xi \in C_\theta \mathcal{M}} \Vert \xi + \nabla f(\theta) \Vert_2  \label{math:proj}
\end{align}
be the Euclidean projection of  $-\nabla f(\theta)$ onto the tangent cone $C_\theta \mathcal{M}$ of $\mathcal{M}$.
The resulting flow $\dot \theta = \Pi_{C_\theta \mathcal{M}}(-\nabla f(\theta))$ may have discontinuities at the boundary $\partial \mathcal{M}$. Following \cites[Secs.~3 and 4.2]{jongen2003constrained}[Sec.~4]{shikhman2009constrained}, a smooth alternative to \eqref{math:proj} is given by $-\nabla_\mathcal{M} f(\theta) = -M(\theta) \nabla f(\theta)$, where}
\begin{align}
   M(\theta) &= I_n - J_g^\top(\theta)\, F(\theta)^{-1} J_g(\theta), \label{math:projected_grad} \\
   F(\theta) &= 2\operatorname{diag}\!\left(g(\theta)\right) + J_g(\theta) J_g^\top(\theta). \label{math:f_x}
\end{align}
%and $J_g(\theta)$ is the Jacobian of $g$ at $\theta$.

% The projection of $\nabla f({\theta })$ onto the tangent cone  $C_{{\theta }} \mathcal{M}$ of $\mathcal{M}$ at ${\theta }\in\mathcal{M}$, denoted by $\nabla_\mathcal{M} f({\theta })$, is the 
% unique solution to
% \begin{align}
%    \min_\xi \Vert \xi - \nabla f({\theta }) \Vert_2 \quad \text{s.t.} \quad \xi \in C_{{\theta }} \mathcal{M}. \label{math:proj}
% \end{align}
% It is explicitly given by $\nabla_\mathcal{M} f({\theta })=M({\theta })\nabla f({\theta })$ \cite{jongen2003constrained}, where 
%    \begin{align}
%       M({\theta }) &= I_n - J_g^\top({\theta }) F({\theta })^{-1} J_g({\theta }), \label{math:projected_grad} \\
%       F({\theta }) &= 2 \operatorname{diag }\left(g\left({\theta }\right)\right) + J_g({\theta }) J_g^\top({\theta }). \label{math:f_x}
%    \end{align}
   
\begin{lemma}
   {
   \label{lemma:5}
   Let $\mathcal{M}=\{{\theta } \in\mathbb{R}^n \mid g({\theta })\geq 0\}$, where $g:\mathbb{R}^n \to \mathbb{R}^l$ is continuously differentiable and LICQ holds at every $\theta\in\mathcal{M}$.
   %Consider the matrices \eqref{math:projected_grad} and \eqref{math:f_x}.
   % Define  \begin{align}
   %    F(x) &= 2 \operatorname{diag }\left(g\left(x\right)\right) + J_g(x) J_g^\top(x), \\
   %    M(x) &= I_n - J_g^\top(x) F(x)^{-1} J_g(x).
   % \end{align}
   Then $F({\theta })\succ 0$ for all ${\theta }\in\mathcal{M}$, and $M({\theta })\succ 0$ if ${\theta }\in \operatorname{int }(\mathcal{M})$ and $M({\theta })\succeq 0$ if ${\theta }\in \partial \mathcal{M}$.}
\end{lemma}
{The proof is provided in Appendix~\ref{app:lemma_5}.} The dynamics of the projected gradient flow is given by 
\begin{align}
   {\dot \theta(t)} \!= \! - \nabla_\mathcal{M} f({\theta (t)}) \!=  \!- M({\theta (t)}) \nabla f({\theta (t)}), \, \theta(0) \!\in \! \mathcal{M}.   \label{math:proj_gradient} 
\end{align}
\begin{proposition}[\texorpdfstring{\cite[Prop. 3.2]{jongen2003constrained}}{}]
   \label{prop:invariant}
   The vector field $ {\theta } \mapsto -\nabla_\mathcal{M} f({\theta })$ is $C^\infty$ smooth on $\mathcal{M}$, and it induces a smooth trajectory ${\theta }(t)$ with $\operatorname{int }(\mathcal{M})$ and $\partial \mathcal{M}$ as invariant manifolds.
\end{proposition}

Problem \eqref{math:opt1} can be locally solved using \eqref{math:proj_gradient}, given that the initial state ${\theta }(0)$ and the optimal solution ${\theta^* }$ reside within the same subset $\partial \mathcal{M}$ or $\operatorname{int }(\mathcal{M})$ {\cite[Thm.~4.2(iii)]{shikhman2009constrained}}.

\section{Problem Formulation and Approach}
\label{sec:2}
In this section, we define the LPV system, the dynamic state-feedback controller, and the closed-loop system. Additionally, we state the main objective of the controller design. 

The polytopic LPV system $\left(A(\rho(t)),B\right)$ considered in this work obeys the dynamics
\begin{align}
   \dot x(t)  = A\left(\rho(t)\right) x(t) + B u(t),\quad x(0) = x_0, \label{math:open_loop}
\end{align}
where $\rho(t)\in\mathcal{F}_{\mathcal{P}}^\infty$, $\mathcal{P}$ is a convex and compact polytope,
$A: \mathbb{R}^p \to \mathbb{R}^{n\times n}$ is a continuous function affine in $\rho(t)$ (see~\eqref{math:affine}), and $B\in\mathbb{R}^{n\times m}$ is constant. 
%The exogenous parameter $\rho$ is assumed to be exactly measurable in real-time. 
% The exogenous parameter $\rho$ is assumed to be measurable.
The dynamic state-feedback controller takes the form $u(t) = - K(t)x(t)$, where the time-varying gain $K(t)$ evolves according to an ODE.
{Henceforth, explicit time dependency of variables is omitted for readability when clear from the context. The notation $\rho(t)$ is preserved to emphasize time-varying parameter trajectories as opposed to the constant $\rho$.}

\begin{objective}
   \label{obj:1}
Design a dynamic state-feedback controller that ensures asymptotic stability of the system \eqref{math:open_loop} under arbitrarily fast parameter variations $\rho(t)\in\mathcal{F}_{\mathcal{P}}^\infty$ while simultaneously minimizing the LQR cost. 
Specifically, the trajectory $K(t)$ of the controller state should converge to the optimal LQR feedback gain for constant parameter trajectories $\rho(t)=\rho \in\mathcal{P}$.
\end{objective}

 To ensure that the infinite-horizon LQR optimization problem is well-posed and that the optimal solution yields a stabilizing controller, we make the following assumption, which will hold throughout the remainder of this paper.
%The following assumption ensures the LQR problem is well-posed and its optimal solution yields a stabilizing controller
\begin{assumption}
   \label{ass:1}
   The weighting matrices satisfy $R \succ 0$ and $Q \succeq 0$. Additionally, the system $(A(\rho), B)$ is stabilizable, and the pair $(A(\rho), \sqrt{Q})$ is detectable for all $\rho \in \mathcal{P}$.
   % The system $(A(\rho),B)$ is stabilizable and the pair $(A(\rho),\sqrt{Q})$ is detectable for all $\rho\in\mathcal{P}$. Additionally, $R\succ 0$ and $Q\succeq 0$.  \label{ass:standard}
\end{assumption}

%\subsection{Approach}
In order to achieve the Objective~\ref{obj:1}, we propose the following dynamic controller that adapts the policy gradient flow (see Subsection \ref{subsubsec:policy_gradient}) from \cite{bu2020clqr} to LPV systems and incorporates the projection method (see Subsection \ref{subsec:proj}) from \cite{jongen2003constrained}.
%Preliminaries on the policy gradient method for LQR and the projected gradient flow are presented in Appendices~\ref{subsubsec:policy_gradient} and \ref{subsec:proj}, respectively.%\footnote{Readers are strongly encouraged to consult the appendix for essential preliminaries, ensuring the following formulas do not appear out of context.}

\begin{definition}[Dynamic state-feedback controller]
   \label{def:dyn}
   Let $\mathcal{C}\subseteq \mathbb{R}^{mn}$ be the hyperrectangle
   \begin{align}
      \label{math:C}
      \mathcal{C} &  = \prod_{i=1}^{mn} [\underline{C}_i, \overline{C}_i]  = \{ { \operatorname{vec}\left(K \right)\in\mathbb{R}^{m n} }\mid g(K)\geq 0\},
   \end{align} 
   where $\underline{C}_i, \overline{C}_i\in\mathbb{R}$, $ \underline{C}_i\! < \! \overline{C}_i$, $g(K) = \left[ g_1(K)\,   \dots \, g_{2nm}(K)\right]^\top$,
   \begin{align} 
      \label{math:ineq}
      g_i(K)& \! = \! \begin{cases}
        K_i - \underline{C}_i &  \!\!\mathrm{for }~  i=1,\dots,mn,\\
        \overline{C}_{{i-mn}}\! -\! K_{{i-mn}} & \!\! \mathrm{for }~ i=mn\!+\!1,\dots,2mn,
      \end{cases}
   \end{align}
   and $K_i$ is the $i$-th element of $\operatorname{vec}(K)$.
   The dynamic state-feedback controller for the system \eqref{math:open_loop} is defined by
   \begin{subequations}
      \label{math:dyn}
      \begin{align}
         \operatorname{vec}(\dot K ) &= - \alpha \nabla_\mathcal{C} f_K(\rho(t)) , \quad \operatorname{vec}(K(0)) \in \operatorname{int}(\mathcal{C}) , \label{math:dyn_1}\\
         u & = -Kx ,
      \end{align}
   \end{subequations}
   where $\alpha>0$ and $\nabla_{\mathcal{C}}f_K(\rho(t))$ denotes {the smooth alternative (see Subsection~\ref{subsec:proj}) to the projection} of the gradient of the LQR cost of $(A(\rho(t)),B)$ onto the tangent cone
   of $\mathcal{C}$. Specifically,
   \begin{align}
      \nabla_{\mathcal{C}} f_K(\rho(t))&  =  M(K) \operatorname{vec}(2 (RK-B^\top P_K) Y_K ), \label{math:proj_grad}
      %M(K) &= I_n - J_g^\top(K) F(K)^{-1} J_g(K), \label{math:proj_matrxi} \\
      %F(K) &= 2 \operatorname{diag }\left(g\left(K\right)\right) + J_g(K) J_g^\top(K), \label{math:M}
   \end{align} where $P_K$ and $Y_K$ solve the Lyapunov equations \eqref{math:Ly} and \eqref{math:Ly2}\footnote{Note that the matrices $P_K$ and $Y_K$ depend on $\rho(t)$, but this dependence is omitted for readability and explicitly indicated only when necessary.}, respectively, with $A$ replaced by $A(\rho(t))$, and the projection matrix $M(K)$ is as in \eqref{math:projected_grad} evaluated at $\operatorname{vec}(K)$.
   % \begin{align}
   %    0 & = A_K^\top(\rho) P_K + P_K A(\rho)_K + Q + K^\top R K, \label{math:Lyap0}\\
   %    0 & = A(\rho)_K Y_K + Y_K A_K^\top(\rho) + I_n,
   % \end{align}
   % with $A_K(\rho) := A(\rho)-BK$.
\end{definition}
{The projection of the LQR gradient requires only evaluating $M(K)$ (see~\eqref{math:proj_grad}), which entails inverting  $F(K)\in\mathbb{R}^{2mn\times 2mn}$. For the hyperrectangle \eqref{math:C},  $F(K)^{-1}$ admits a closed-form expression for arbitrary dimensions $n$ and $m$, making the projection computationally inexpensive.}
%Note that the matrices $P_K$ and $Y_K$ depends on $\rho(t)$%, but this dependence is omitted for readability explicitly indicated only when necessary.
% {
%    \begin{remark}
    
%       % . The matrix $M(K)$ comprises the evaluation of the Jacobian of the affine $g(K)$ and a inversion of the $2mn \times 2mn$ matrix $F(K)$, making the projection computationally inexpensive.
%    %    The projection of the LQR gradient onto the tangent cone of $\mathcal{C}$ is equal premultiplying the  $mn$ dimensional symmetric matrix $M(K)$ to the gradient. The matrix $M(K)$ evaluation at $K$ requires evaluation of the affine inequality $g(K)$, its Jacobian and the inversion of $2mn$ dimensional matrix 
%    %   requires  $M(K)$. 
%    %    Since $g(K)$ is affine in terms of $K$, $J_g(K)$
%    %    The projected gradient  ,making the projection computationally inexpensive. $M(K)\in\mathbb{R}^{mn\times mn}$ is a symmetric matrix 
%    \end{remark}

\begin{remark}
   \label{remark:LQR_opt}
      The projected gradient flow \eqref{math:dyn} seeks to solve % the optimization problem
      \begin{align}
         \min_{K\in\mathcal{C}} f_K(\rho(t)) = \min_{K\in\mathcal{C}}\operatorname{tr }\left(P_K(\rho(t))\right), \label{math:LPV_prob}
      \end{align}
      which is a modified LQR problem where the feasible set $\mathcal{C}$ may not correspond to the set of all stabilizing feedback gains. Moreover, the optimization problem \eqref{math:LPV_prob} is time-varying due to $\rho(t)\in\mathcal{F}_{\mathcal{P}}^\infty$ and the time-varying minimizer may never be attained by the projected gradient flow~\eqref{math:dyn_1}.
\end{remark}

Two important aspects should be noted. Firstly, the projected gradient flow \eqref{math:dyn_1} is designed to ensure that the trajectory $\operatorname{vec}(K(t))$ remains within the hyperrectangle $\mathcal{C}$. However, the gradient $\nabla f_K$, and consequently the projected gradient  $\nabla_\mathcal{C} f_K$ (see Subsection \ref{subsec:proj}), may be undefined for some $\operatorname{vec}(K ) \in\mathcal{C}$. Furthermore, the projected gradient flow could diverge for certain parameter trajectories $\rho(t)\in\mathcal{F}_{\mathcal{P}}^\infty$ and specific choices of $\mathcal{C}$. Secondly, the optimal LQR feedback gains $ \operatorname{vec}(K^*_\rho ) $ may not lie within the hyperrectangle $\mathcal{C}$ for all $\rho\in\mathcal{P}$. Consequently, the trajectory $K(t)$ generated by \eqref{math:dyn_1} may not converge to the optimal feedback gains $K^*_\rho$ for constant parameter trajectories $\rho(t)=\rho\in\mathcal{P}$. 
%These aspects will be addressed in the following section.
These aspects are analyzed and resolved in the following section.
Before proceeding, we first define the closed-loop system by combining the open-loop system~\eqref{math:open_loop} and the controller \eqref{math:dyn}.

\begin{definition}[Closed-loop system]
   The closed-loop system is 
   \vspace{-0.4cm}
   \begin{subequations}
      \label{math:closed_loop}
      \begin{align}
         \dot x &= (A(\rho(t))- BK)x, \label{math:subsys_1} \\ 
         \operatorname{vec}(\dot K ) &= - \alpha \nabla_\mathcal{C} f_K(\rho(t)), \label{math:subsys_2} 
      \end{align}
   \end{subequations}
   where $A(\rho(t))$ is as in \eqref{math:open_loop} and \eqref{math:subsys_2} is defined in Definition~\ref{def:dyn}.
\end{definition}

\begin{remark}
   \label{remark:not_LPV}
   System \eqref{math:closed_loop} is nonlinear due to the product $Kx$ of the states and the adaptation rule $- \alpha \nabla_\mathcal{C} f_K(\rho(t))$, which is a vector of rational functions in the entries of~$K$. Thus, \eqref{math:closed_loop} is not a generic LPV system as in \eqref{math:LPV} and cannot be analyzed using standard LPV methods.
\end{remark}

\section{Stability and Optimality Analysis}
\label{sec:stability}
% In this section, we address the quadratic stability of the LPV system in Subsection \ref{subsec:quad}, the boundedness and the optimality of the dynamic controller trajectories in Subsection \ref{subsec:bound}, and the asymptotic frozen-time stability of the closed-loop system in Subsection~\ref{subsection:frozen}.
This section addresses the quadratic stability of the LPV system in Subsection~\ref{subsec:quad}, the boundedness and the optimality of the controller trajectories in Subsection~\ref{subsec:bound}, and the asymptotic stability of the frozen-time closed-loop systems in Subsection~\ref{subsection:frozen}. To facilitate readability, Fig.~\ref{fig:4} illustrates the logical dependencies among the contributions, assumptions, propositions, and theorems in this section.

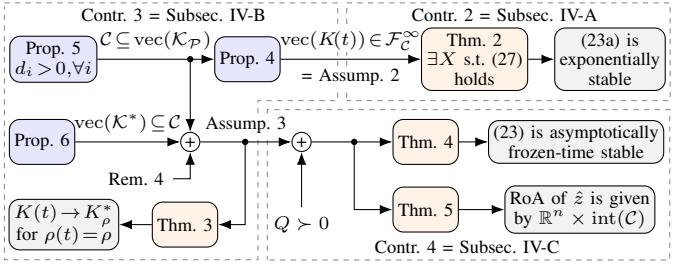
\begin{figure}
   \centering
\begin{tikzpicture}[
   node style/.style={rectangle, draw=black, align=center, minimum height=1.7em, minimum width=2em,rounded corners},
   prop style/.style={rectangle, draw=black, align=center, minimum height=1.7em, minimum width=2em,rounded corners,fill=blue!10,inner sep=1.9pt},
   theo style/.style={rectangle, draw=black, align=center, minimum height=1.7em, minimum width=2em,rounded corners,fill=orange!10,inner sep=1.9pt},
   results style/.style={rectangle, draw=black, align=center, minimum height=1em, minimum width=2em,rounded corners,fill=gray!10,inner sep=1.9pt},
   text style/.style={draw=none, fill=none,align=center},
   plus style/.style={circle, draw=black, inner sep=0.7pt, minimum size=0.6em, font=\scriptsize},
   arrow style/.style={-{Latex}},font=\fontsize{7pt}{7.5pt}\selectfont
 ]

\node[draw=gray!90](contr1) at (6.05,0.1) [rectangle,dash pattern=on 2pt off 2pt,dash phase=0.35pt,minimum width=4.3cm,minimum height=1.3cm] {};
\node[draw=gray!90](contr2) at (5.52,-1.67) [rectangle,dash pattern=on 2pt off 2pt,dash phase=3.2pt,minimum width=5.35cm,minimum height=2.0cm] {};
\draw[draw=gray!90,dash pattern=on 2pt off 2pt,dash phase=0pt] (-0.65,0.75) -- (3.78,0.75) -- (3.78,-0.55) -- (2.72,-0.55) -- (2.72,-2.667)--(-0.65,-2.667) -- (-0.65,0.7);
\node[text style] (T7) at (6,0.59) {Contr. \ref{item:second} = Subsec. \ref{subsec:quad}};
\node[text style] (T8) at (1.6,0.59) {Contr. \ref{item:third} = Subsec. \ref{subsec:bound}};
\node[text style] (T9) at (5.6,-2.49) {Contr. \ref{item:fourth} = Subsec. \ref{subsection:frozen}};

 % Nodes
 \node[prop style] (A) at (0,0) {Prop. \ref{prop:2} \\
 $d_i \!> \!0,\! \forall i$ %\\
 %$\forall i $%\!  = \! 1, \kern-0.1em .\kern-0.1em .\kern-0.1em. ,n$
 };
 \node[prop style] (B) at (2.58,0) {Prop. \ref{prop:bound}};
 \node[theo style] (C) at (5.6,0) {Thm. \ref{theo:1} \\
 $\exists X~\text{s.t.}~$\eqref{math:quad_2} \\
 holds
 };
 \node[results style] (T1) at (7.4,0) {
\eqref{math:subsys_1} is \\ exponentially \\ stable
 };
 \node[prop style] (D) at (-0.15,-1.1) {Prop. \ref{prop:1}};
 \node[plus style] (P1) at (1.82,-1.1){+};
 \node[plus style] (P2) at (3.4,-1.1){+};
 \node[theo style] (E) at (4.95,-1.1) {Thm. \ref{theo:froze}};
 \node[theo style] (F) at (4.95,-2.0) {Thm. \ref{theo:region}};
 \node[results style] (T2) at (6.95,-1.1) {
   \eqref{math:closed_loop} is asymptotically \\ frozen-time stable
       };
   \node[text style] (T3) at (3.4,-2.2) { $Q\succ 0$ \\ ~~~ {$Y_K \! \leftarrow  \! Y_{K, x}$}};
   \node[results style] (T4) at (7,-2) {
       RoA of $\hat{z}$ is given \\ by $\mathbb{R}^n\times \operatorname{int}( \mathcal{C})$
           };

   \node[text style] (T5) at (1.1,-1.59) {Rem. \ref{rem:strict}};
   \node[results style] (G) at (0.14,-2.2) {
      $K(t) \! \to \! K_\rho^*$ \\ for $\rho(t)\!=\!\rho$
           };
   \node[theo style] (H) at (1.75,-2.2) {Thm. \ref{theo:3}};

 % Arrows
 \draw[arrow style] (A) -- node[midway, above] {$\mathcal{C} \!\subseteq  \!\operatorname{vec}( \mathcal{K}_\mathcal{P} )$} (B);
  \fill[white] (3.4,0.1) rectangle +(1cm,0.32cm);
   \fill[white] (3.4,-0.38) rectangle +(1cm,0.285cm);

 \draw[arrow style] (B) -- node[midway, above] {$ \operatorname{vec}(K\!(t) )\!\in \!\mathcal{F}_\mathcal{C}^\infty$} node[midway, below] {= Assump. \ref{ass:bound}}(C);
% \draw[arrow style] (B) -- (C);
% \node[fill=white, inner sep=0pt] at (4,0.2) {$ \operatorname{vec}(K\!(t) )\!\in \!\mathcal{F}_\mathcal{C}^\infty$};
 \draw[arrow style] (C) -- (T1);
 \draw[arrow style] (D) -- node[midway, above] {$\operatorname{vec}( \mathcal{K}^* )\!\subseteq \!\mathcal{C} $} (P1);
 \draw[arrow style] (1.82,0) -- (P1);
    \fill[white] (2,-1.05) rectangle +(1cm,0.285cm);
 \draw[arrow style] (P1) --  node[midway, above] {Assump. \ref{ass:opt}} (P2);
 \draw[arrow style] (P2) -- (E);
 \draw[arrow style] (4,-1.1) |- (F.west);
 \draw[arrow style] (E) -- (T2);
 \draw[arrow style] (F) -- (T4);
 \draw[arrow style] (T3) -- (P2);
 \draw[arrow style] (T5) -| (P1);
 \draw[arrow style] (2.55,-1.1) |- (H);
 \draw[arrow style] (H) -- (G);
 \fill (1.82, 0) circle[radius=1.2pt];
 \fill (2.55, -1.1) circle[radius=1.2pt];
 \fill (4,-1.1) circle[radius=1.2pt];
\end{tikzpicture}
   \caption{Structured overview of the logical dependencies among the contributions, assumptions, propositions, and theorems of Section \ref{sec:stability}. The subsections are separated by Assumptions~\ref{ass:bound} and~\ref{ass:opt}.}
   \label{fig:4}
\end{figure}

\subsection{Quadratic Stability}
\label{subsec:quad}
In this subsection, we derive sufficient conditions for the quadratic stability of subsystem~\eqref{math:subsys_1}.

\begin{assumption}
   \label{ass:bound}
   The trajectory $\operatorname{vec}(K(t))$ of the dynamic controller \eqref{math:dyn_1} stays in $\mathcal{C}$
   for $\rho(t)\in\mathcal{F}_{\mathcal{P}}^\infty$. 
\end{assumption}
Sufficient conditions ensuring that Assumption~\ref{ass:bound} holds are presented in the subsequent subsection.
\begin{proposition}
   \label{prop:quad_stab}
   Let Assumption \ref{ass:bound} hold. If there exists a matrix $X\succ 0$ such that 
   \begin{align}
      \left(A(\rho) - BK\right)^\top X +X  \left(A(\rho) - BK\right) \prec 0 \label{math:quad_1}
   \end{align}
   holds for all $\rho\in\mathcal{P}$ and all $ \operatorname{vec}(K )\in\mathcal{C}$, then the subsystem~\eqref{math:subsys_1} is quadratically stable for all $\rho(t)\in\mathcal{F}_{\mathcal{P}}^\infty$.
\end{proposition}

\begin{proof}
   Since the trajectories of the dynamic controller remain in $\mathcal{C}$, $ \operatorname{vec}(K(t) )\in \mathcal{F}_{\mathcal{C}}^\infty$ holds and the feedback gain $K$ can be interpreted as a varying parameter. Let $\tilde{\rho }^\top = \begin{bmatrix} \rho^\top & \operatorname{vec}(K)^\top\end{bmatrix}$, ${A}_K(\tilde{\rho})= A(\rho)-B K$ and $\tilde{\mathcal{P}}=\mathcal{P}\times \mathcal{C}$. The subsystem \eqref{math:subsys_1} can be represented as a generic LPV system \eqref{math:LPV} with ${A}_K(\tilde{\rho})$ and $\tilde{\rho}\in\mathcal{F}_{\tilde{\mathcal{P}}}^\infty$. 
   Hence, the LMI conditions of Proposition \ref{prop:quad} apply, and \eqref{math:subsys_1} is quadratically stable for all $\rho(t)\in\mathcal{F}_{\mathcal{P}}^\infty$ if the LMI \eqref{math:quad_1} holds for all $\rho\in\mathcal{P}$ and $ \operatorname{vec}(K )\in\mathcal{C}$.
\end{proof}

Note that \eqref{math:quad_1} must hold for infinitely many values of $\rho\in\mathcal{P}$. To make \eqref{math:quad_1} tractable, we reformulate it into a finite set of LMIs, requiring the following lemma.

\begin{lemma}
   \label{lemma:poly}
   Let Assumption \ref{ass:bound} hold. Then, the subsystem \eqref{math:subsys_1} is a polytopic LPV system and can be represented as a time-varying convex combination of $N$ LTI systems $ \dot x = \left(\sum_{i=1}^N \lambda_i(t) \tilde{A}_{K,i} \right)x, \quad \lambda(t) \in \Lambda_{N}$.
   % \begin{align}
   %    \dot x = \left(\sum_{i=1}^N \lambda_i(t) \tilde{A}_{K,i} \right)x, \quad \lambda(t) \in \Lambda_{N}.
   % \end{align}
\end{lemma}
%\vspace{0.15cm}
\begin{proof}
As in the proof of Proposition \ref{prop:quad_stab}, we interpret the controller as a varying parameter. Since both $\mathcal{P}$ and $\mathcal{C}$ are polytopes and $A_K(\tilde{\rho}) = A(\rho) - BK$ is affine in $\rho$ and $K$, the subsystem \eqref{math:subsys_1} is a polytopic LPV system, representable as a time-varying convex combination of LTI systems (see \eqref{math:LTI_comb}).
%Next, we derive the matrices $\tilde{A}_{K,i}$.
Let $\tilde{\mathcal{P}}=\mathcal{P} \times \mathcal{C}$.
Then, $\tilde{\mathcal{P}}$ has $N= \vert \operatorname{vert }(\mathcal{P}) \vert \cdot 2^{mn}$ vertices denoted by $v_i = \begin{bmatrix} {v_i^{\mathcal{P}}}^\top & {v_i^{\mathcal{C}}}^\top \end{bmatrix}^\top, i=1,\dots,N$. 
Using the affine structure of $A(\rho(t))$ (see \eqref{math:affine}), 
the corresponding LTI matrices are given by  
$ \tilde{A}_{K,i}= A_0 + \sum_{j=1}^{p} \left(v_i^{\mathcal{P}}\right)_j A_j - B  \operatorname{vec}^{-1}_{m\times n}\left(v_i^{\mathcal{C}}\right)$
   % \begin{align}
   %    \tilde{A}_{K,i}= A_0 + \sum_{j=1}^{p} \left(v_i^{\mathcal{P}}\right)_j A_j - B  \operatorname{vec}^{-1}_{m\times n}\left(v_i^{\mathcal{C}}\right),
   % \end{align}
   where $\left( \cdot \right)_j$ denotes the $j$-th element of $v_i^{\mathcal{P}}$, $ \operatorname{vec}^{-1}_{m\times n}$ denotes the inverse vectorization operator\footnote{The inverse operator $\operatorname{vec}^{-1}_{m\times n}:\mathbb{R}^{mn}\to \mathbb{R}^{m\times n}$ can be defined as $v \mapsto \operatorname{vec}_{m\times n}^{-1}(v )=\left( \operatorname{vec}^\top(I_n ) \otimes I_m\right)(I_n \otimes v)$.} and $i=1,\dots,N$.
\end{proof}

%With Lemma \ref{lemma:poly}, we reformulate the condition \eqref{math:quad_1} to a finite set of LMIs in the following theorem.

Lemma \ref{lemma:poly} enables the reformulation of the condition \eqref{math:quad_1} as a finite set of LMIs, as stated in the following theorem.

\begin{theorem}
   \label{theo:1}
   Let Assumption \ref{ass:bound} hold.
   The subsystem \eqref{math:subsys_1} is quadratically stable if and only if there exists a matrix $X\succ 0$ such that
   \begin{align}
      \tilde{A}_{K,i}^\top X + X \tilde{A}_{K,i} \prec 0, \quad \forall i=1,\dots,N. \label{math:quad_2}
   \end{align}
   %holds for all $i=1,\dots,N$.
\end{theorem}
\begin{proof}
   Under the given assumption, the subsystem \eqref{math:subsys_1} can be interpreted as a polytopic LPV system (see Lemma \ref{lemma:poly}). Hence, \eqref{math:quad_2} is equivalent to the condition of Theorem \ref{theo:poly_quad}.
\end{proof}

Thus far, we have shown that the subsystem~\eqref{math:subsys_1} is 
quadratically stable, and consequently, exponentially stable (see Proposition~\ref{prop:quad}) for arbitrarily fast parameter variations, provided that Assumption \ref{ass:bound} and condition~\eqref{math:quad_2} hold. 

\subsection{Boundedness and Optimality of the Dynamic Controller}
\label{subsec:bound}
This subsection addresses the construction of a hyperrectangle $\mathcal{C}$ that satisfies Assumption~\ref{ass:bound} and contains the optimal LQR gains $K^*_\rho$ for all $\rho\in\mathcal{P}$. We further establish conditions under which the trajectory $K(t)$ of the dynamic controller \eqref{math:dyn_1} converges to the optimal feedback gain $K_\rho^*$ for constant parameter trajectories $\rho(t)=\rho\in\mathcal{P}$. 

The set of Hurwitz stable feedback gains of the system $\left(A(\rho),B\right)$ with a given parameter $\rho\in\mathcal{P}$ is denoted as 
\begin{align}
   \mathcal{K}_{\rho} = \{ K\in\mathbb{R}^{m\times n }: \operatorname{Re}\left(\lambda_i(A(\rho)-BK)\right) < 0 ~ \forall i \}
\end{align}
and is referred to as the stability region of the system $(A(\rho),B)$.
The intersection of the sets $\mathcal{K}_\rho$ for all $\rho\in\mathcal{P}$ is defined as $\mathcal{K}_{\mathcal{P}} = \bigcap_{\rho\in\mathcal{P}} \mathcal{K}_\rho$.\footnote{If the parameter variation is sufficiently large, $\mathcal{K}_{\mathcal{P}}$ may be an empty set.}
%A sufficient condition to ensure that Assumption \ref{ass:bound} holds is provided next.
% \vspace{0.15cm}
\begin{proposition}
   \label{prop:bound}
   If $\mathcal{C}\subseteq \operatorname{vec}\left(\mathcal{K}_{\mathcal{P}}\right)$ holds, then the trajectory $\operatorname{vec}(K(t) )$ of the dynamic controller \eqref{math:dyn_1} remains in $\operatorname{int}( \mathcal{C})$ for $\rho(t)\in\mathcal{F}_{\mathcal{P}}^\infty$ and is continuous, i.e., $ \operatorname{vec}(K(t) ) \in\mathcal{F}_{\mathcal{C}}^\mathcal{V}$.
\end{proposition}
\begin{proof}
Since $\mathcal{C}$ is convex, the tangent cone $C_{ {\operatorname{vec}\left(K \right)} } \mathcal{C}$ and the projection onto it are well-defined, unique and independent of $\rho(t)$.
   We first verify that $\nabla f_K(\rho)$ is well-defined over $\mathcal{C}$ for all $\rho \in \mathcal{P}$. 
   By \cite[Prop. 3.2]{bu2020clqr}, the cost function $f_K(\rho)$ is $C^\infty$ smooth over $\mathcal{K}_\rho$, ensuring the existence of the gradient $\nabla f_K(\rho)$ for all $\operatorname{vec}(K )\in \mathcal{C}\subseteq \operatorname{vec}(\mathcal{K}_{\mathcal{P}})\subseteq \operatorname{vec }(\mathcal{K}_\rho)  $ and $\rho \in \mathcal{P}$. Since $ \operatorname{vec}(K(0))\in \operatorname{int}(\mathcal{C})$, the projected gradient \eqref{math:dyn_1} ensures that $\operatorname{vec}(K(t) )$ remains in $\operatorname{int}(\mathcal{C})$ due to the invariance of $\operatorname{int }(\mathcal{C})$ (see Proposition \ref{prop:invariant}). 
   Although $\nabla f_K(\rho(t))$ may exhibit discontinuities due to $\rho(t) \in \mathcal{F}_{\mathcal{P}}^\infty$, the trajectory 
   $\operatorname{vec}(K(t) ) = \operatorname{vec}(K(0)) - \alpha \int_0^t \nabla_\mathcal{C} f_K(\rho(\tau)) \mathrm{d}\tau$
   % \begin{align}
   %    \operatorname{vec}(K(t) ) = \operatorname{vec}(K(0)) - \alpha \int_0^t \nabla_\mathcal{C} f_K(\rho(\tau)) \mathrm{d}\tau
   % \end{align} 
   remains continuous,  
   since the integrand is piecewise continuous with finitely many discontinuities  and
bounded.%, as $\nabla_\mathcal{C} f_K(\rho)$ is continuous on the compact set
%$\mathcal{C}\times\mathcal{P}$.
   % as the possible discontinuities of $\rho(t)$ are finite, given the compactness of $\mathcal{P}$.
\end{proof}

Since the subsystem \eqref{math:subsys_1} is linear in the state $x$, the Routh-Hurwitz stability criterion (see Appendix \ref{subsec:routh}) can be applied to determine the stability region $\mathcal{K}_\rho$ for a given $\rho\in\mathcal{P}$.

\begin{lemma}
   \label{lemma:stable_poly}
   The stability region $\mathcal{K}_\rho$ is characterized by $n$ polynomial inequalities $f_i(K,\rho) > 0, i=1,\dots,n$, where each $f_i(K,\rho)$ is a polynomial in the entries of $K$ and $\rho$. Formally, $ \mathcal{K}_\rho = \{ K \in \mathbb{R}^{m \times n} : f_i(K,\rho) > 0, ~ i=1,\dots,n \}.$
   % \begin{align}
   %    \label{math:stab_pol}
   %    \mathcal{K}_\rho = \{ K \in \mathbb{R}^{m \times n} : f_i(K,\rho) > 0, \quad i=1,\dots,n \}.
   % \end{align}
 \end{lemma}
\begin{proof}
   The subsystem \eqref{math:subsys_1} is asymptotically stable if and only if the roots of the characteristic polynomial $p(s) = \operatorname{det }(sI-(A(\rho)-BK))$ have negative real parts. This condition holds if and only if all entries in the first column of the Routh table~\eqref{math:Routh} of $p(s)$ are positive, i.e., $c_1,R_{3,1},\dots,R_{n,1},c_n > 0$.
   The entries $R_{3,1},\dots,R_{n,1}$ are rational functions, as they involve divisions by $c_1,R_{3,1},\dots,R_{n-1,1}$ which have to be strictly positive such that the criterion holds.
   %Since it is allowed to multiply any row of the Routh table by a positive constant,
   Since multiplying any row of the Routh table by a positive constant does not affect the stability criterion,   
   a modified Routh table can be obtained by using the recursion $\tilde{R}_{i,j} = \tilde R_{i-1,1} \tilde R_{i-2,j+1}- \tilde R_{i-2,1} \tilde R_{i-1,j+1}$ for $3\leq i \leq n$, where the division by $R_{i-1,1}$ is omitted. 
   This results in $n$ polynomials $f_1(K,\rho)=c_1,f_2(K,\rho)=\tilde R_{3,1},\dots, f_{n-1}(K,\rho)=\tilde R_{n,1}, f_n(K,\rho) = c_n$ that must each be strictly positive.
\end{proof}

Next, we address how the condition $\mathcal{C} \subseteq \operatorname{vec}\left(\mathcal{K}_{\mathcal{P}}\right)$ of Proposition~\ref{prop:bound} can be verified, as it ensures that Assumption \ref{ass:bound} holds, which is required for quadratic stability (see Theorem~\ref{theo:1}).

\begin{proposition}
   \label{prop:2}
   The condition $\mathcal{C}\subseteq \operatorname{vec}( \mathcal{K}_{\mathcal{P}})$ holds, if, for each $i=1,\dots,n$,  the  optimization problem
   \begin{subequations}
      \label{math:opt2}
   \begin{align}
      d_i= \min_{K_{\mathcal{P}},K_\mathcal{C},\rho}~ & \Vert K_\mathcal{C} - K_\mathcal{P} \Vert_F \\
       \text{s.t.}~ & f_i(K_\mathcal{P},\rho) \leq 0, \label{math:c1}\\
       & \operatorname{vec}( K_\mathcal{C})\in\mathcal{C}, ~\rho \in\mathcal{P},  \label{math:c3}
      \end{align}
   \end{subequations}
   is feasible and admits an optimal value  $d_i>0$.
\end{proposition}

\begin{proof}
   The optimal value of \eqref{math:opt2} represents the minimum distance between the boundary defined by $f_i(K_{\mathcal{P}},\rho) = 0$ and the set $\mathcal{C}$ over all $\rho \in \mathcal{P}$. If $d_i > 0$, then \eqref{math:c1} ensures that $\mathcal{C}$ lies entirely in the region where $f_i(K_{\mathcal{P}},\rho) > 0$. Conversely, if $d_i = 0$, then $\mathcal{C}$ intersects the boundary, implying $\mathcal{C} \not\subseteq \operatorname{vec}(\mathcal{K}_{\mathcal{P}})$. Thus, if all $d_i$ are strictly positive, 
    $\mathcal{C}$ is fully contained in the region where $f_i(K_\mathcal{P},\rho) > 0$ for all $i=1,\dots,n$ and $\rho\in\mathcal{P}$. Since this region characterizes $\mathcal{K}_{\mathcal{P}}$ (see Lemma~\ref{lemma:stable_poly}), the condition $\mathcal{C}\subseteq\operatorname{vec}( \mathcal{K}_{\mathcal{P}})$ holds.
\end{proof}
% The constraints \eqref{math:c2} and \eqref{math:c3} are convex inequalities, but for $n > 2$, the optimization problems \eqref{math:opt2} are generally nonconvex due to \eqref{math:c1}, which is a polynomial inequality in the entries of $K$ and $\rho$. 
Note that, while the constraints \eqref{math:c3} are convex, the problem \eqref{math:opt2} is  generally nonconvex for $n > 1$ due to \eqref{math:c1}, which is a polynomial inequality in the entries of $K$ and $\rho$.

The union of the optimal feedback gains $K^*_\rho$ for all $\rho\in\mathcal{P}$ is denoted by $\mathcal{K}^* = \bigcup_{\rho\in\mathcal{P}} \{ K^*_\rho \}$.
%The optimal feedback gain for the system $(A(\rho),B)$ is denoted by $K^*_\rho$, and the union of these feedbacks $K^*_\rho$ for all $\rho\in\mathcal{P}$ is $\mathcal{K}^* = \bigcup_{\rho\in\mathcal{P}} \{ K^*_\rho \}$. 
% Since $\mathcal{K}^*$ is implicitly defined by the maximal solutions of the CARE for all $\rho\in\mathcal{P}$, we use an over-approximation of $\mathcal{K}^*$ to simplify its characterization, justified by the following properties of $\mathcal{K}^*$.
The set $\mathcal{K}^*$ is challenging to characterize exactly, as it is implicitly defined by the maximal solutions of the CARE for all $\rho \in \mathcal{P}$. We avoid this issue by utilizing an over-approximation of $\mathcal{K}^*$, which is justified by the following properties of $\mathcal{K}^*$.

\begin{lemma}
   \label{lemma:K_star}
  The set $\mathcal{K}^*$ is compact and path-connected, and the mapping $\rho \mapsto K^*_\rho$ is continuous.
\end{lemma}
\begin{proof}
   The composition
    $   \rho \mapsto A(\rho) \mapsto P^*(\rho) \mapsto K^*_\rho$
   % \begin{align}
   %    \rho \mapsto A(\rho) \mapsto P^*(\rho) \mapsto K^*_\rho
   % \end{align}
   consists of continuous maps. Since $A(\rho)$ is linear in $\rho$, it is continuous in $\rho$.  
   %The matrix $P^*(a)$ is the positive definite solution of the CARE. 
   By \cite[Thm. 11.2.1]{lancaster1995algebraic}, the maximal solution $P^*(\rho)$ of the CARE depends continuously on $A(\rho)$.
 The function $K^*_\rho=R^{-1}B^\top P^*(\rho)$ inherits the continuity, because  $R\succ 0$ and $B$ are constant. Hence, the mapping $\rho\mapsto K^*_\rho$ is continuous.
   Since $\mathcal{P}$ is compact and path-connected by definition, the image $K^*_\mathcal{P}=\mathcal{K}^*$ of the continuous function $\rho\mapsto K^*_\rho$ is also compact and path-connected \cite[Thm. 23.5 and Thm. 26.5]{munkres2014topology}.
\end{proof}

\begin{proposition}
   \label{prop:1}
   The set $\mathcal{K}^*$ can be over-approximated by a  hyperrectangle $\mathcal{C}\subset \mathbb{R}^{mn}$, i.e., $ \operatorname{vec}(\mathcal{K}^*) \subseteq \mathcal{C}$, defined as $ \mathcal{C} = \prod_{i=1}^{mn} [\underline{K}_i, \overline{K}_i],$
   % \begin{align}
   %    \mathcal{C} = \prod_{i=1}^{mn} [\underline{K}_i, \overline{K}_i],
   % \end{align}
   where $\underline{K}_i,\overline{K}_i\in\mathbb{R}$. The bounds $\underline{K}_i$ and $\overline{K}_i$ of the $i$-th element of $\operatorname{vec}(K)$ can be determined by solving the following optimization problem for each $i=1,\dots,mn$
   \begin{subequations}
      \label{math:opt_C}   
      \begin{align}
      \min_{K,P,\rho} & \pm K_i \\
      \text{s.t. } & A(\rho)^\top P + P A(\rho) - PBR^{-1}B^\top P + Q = 0, \label{math:CARE_opt} \\
      & K =  R^{-1} B^\top P, \\
      & \rho\in\mathcal{P},~ P\succeq 0.
   \end{align}
\end{subequations}
\end{proposition}

\begin{proof}
   By Lemma \ref{lemma:K_star}, the set $\mathcal{K}^*$ is bounded, so $\mathcal{C}\supseteq \mathcal{K}^*$ exists. The optimal values of the optimization problems \eqref{math:opt_C} provide tight bounds on each element of $K$ by satisfying the CARE \eqref{math:CARE_opt} and considering all admissible parameters $\rho\in\mathcal{P}$.
\end{proof}

\begin{remark}
    The problem \eqref{math:opt_C} is nonconvex due to the CARE \eqref{math:CARE_opt} and is challenging to solve, especially for high-dimensional systems. {
   In practice, one may sample $\rho\in\mathcal{P}$ and solve the corresponding CARE to approximate $\mathcal{K}^*$ for constructing $\mathcal{C}$, though this does not guarantee $  \operatorname{vec}(\mathcal{K}^*) \subseteq \mathcal{C}$.}
\end{remark}

% \begin{remark}
%    The optimization problem \eqref{math:opt_C} is nonconvex due to the CARE \eqref{math:CARE_opt} and is challenging to solve, especially for high-dimensional systems.
%    However, in the following, we assume that the unique global solution of \eqref{math:opt_C} can be obtained.
% \end{remark}

% \begin{remark}
%    {In practice, particularly for higher-dimensional systems, it is more tractable to sample $\rho\in\mathcal{P}$ and solve the corresponding CARE to obtain an approximation of $\mathcal{K}^*$ for constructing $\mathcal{C}$. This sampling-based approach does not provide a formal guarantee that $\operatorname{vec}(\mathcal{K}^*) \subset \mathcal{C}$ holds.}
% \end{remark}

\begin{assumption}
   \label{ass:opt}
   Let $ \operatorname{vec}( \mathcal{K}^*) \subseteq \operatorname{int}( \mathcal{C} )$ and $\mathcal{C}\subseteq \operatorname{vec}( \mathcal{K}_{\mathcal{P}})$ hold. 
\end{assumption}
%Assumption \ref{ass:opt} requires that the optimal feedback gains lie strictly inside $\mathcal{C}$, which itself is contained in the intersection of all stabilizing regions.
Note that $\mathcal{C} \subseteq \operatorname{vec}(\mathcal{K}_{\mathcal{P}})$ ensures quadratic stability of subsystem \eqref{math:subsys_1} (see Proposition ~\ref{prop:bound}, Theorem ~\ref{theo:1}). Assumption~\ref{ass:opt} further requires that the optimal feedback gains lie strictly within $\mathcal{C}$, which represents a more conservative requirement.

\begin{remark}
   \label{rem:strict}
   The condition $ \operatorname{vec}( \mathcal{K}^*) \subseteq \operatorname{int}( \mathcal{C} )$ could practically be achieved by adding a small margin $\varepsilon>0$ to the $2mn$ inequalities \eqref{math:ineq} of $\mathcal{C}$, i.e., $\underline{K}_i - \varepsilon \leq K_i \leq \overline{K}_i + \varepsilon$, where $\overline{K}_i,\underline{K}_i$ are as in Proposition \ref{prop:1}.
\end{remark}

\begin{theorem}
    \label{theo:3}
  Let Assumption \ref{ass:opt} and $\rho(t)=\rho\in\mathcal{P}$ hold. 
   Then, the projected gradient flow \eqref{math:dyn_1} of the dynamic controller solves the LQR optimization problem 
   \begin{align}
      \label{math:opt_2}
      %\min_{K\in\mathcal{C}} f_K(\rho)=
      \min_{K\in\mathcal{K}_\rho} f_K(\rho)
   \end{align}
   and its trajectory $K(t)$ converges to the optimal feedback $ K^*_\rho$.
\end{theorem}

\begin{proof}
   The cost $f_K(\rho)$ has a single stationary point $K_\rho^*$ in  $\mathcal{K}_\rho$.    
   Since  $\operatorname{vec}( \mathcal{K}^*) \subseteq \operatorname{int}( \mathcal{C} )$ and $\mathcal{C}\subseteq \operatorname{vec}( \mathcal{K}_{\mathcal{P}})$, $K^*_\rho$ is the only stationary point within $\operatorname{int }(\mathcal{C})$.
   To confirm that \eqref{math:dyn_1} solves \eqref{math:opt_2}, we need to verify the three assumptions mentioned in Subsection \ref{subsec:proj}.
   By definition \eqref{math:C}, $\mathcal{C}$ {is compact,  connected and satisfies the LICQ.} 
   % Since $\operatorname{vec}(K(t) )$ evolves within $\mathcal{C}$ (see Proposition \ref{prop:bound}), no constraints are active, trivially satisfying the LICQ.
   Within $\operatorname{int}(\mathcal{C})$, there is only one KKT point $K^*_\rho$ due to the single stationary point. 
   The Hessian of the Lagrangian function at $ K^*_\rho$ corresponds to the Hessian of $f_K(\rho)$ which is nondegenerate (see \cite[Prop.~3.4]{bu2020clqr}), ensuring the nondegeneracy of the KKT point.
   With all assumptions fulfilled, %the projected gradient flow \eqref{math:dyn_1} solves \eqref{math:opt_2} and
    $K(t)$ converges to $K^*_\rho$.
\end{proof}

\begin{remark}
   \label{rem:5}
   Let Assumption \ref{ass:opt} and $\rho(t)\in \mathcal{F}_{\mathcal{P}}^{\text{pc}}$ hold.
   If the learning rate $\alpha$ of the dynamic controller \eqref{math:dyn_1} is sufficiently large, the trajectory $K(t)$ converges quickly to the optimal gain $K^*_\rho$ in each time interval with a constant parameter. Once $K(t)$ converges, the states $x$ of subsystem \eqref{math:subsys_1} are regulated by the optimal feedback gain $K^*_\rho$, ensuring optimal stabilization with respect to the LQR cost most of the time. 
\end{remark}
\begin{remark}
   {In contrast to \cite{Courties1999,scherer1996mixed}, which address the $H_2$ performance of LPV systems under output feedback, the presented approach minimizes the LQR performance under state feedback. The proposed dynamic controller \eqref{math:dyn} can be interpreted as a gain-scheduled controller, as its feedback gain varies with the scheduling signal $\rho(t)$.}
\end{remark}
\vspace{-0.4cm}
\subsection{Stability of Frozen-Time Closed-Loop Systems}
\label{subsection:frozen}

Previously, we analyzed the stability of subsystem \eqref{math:subsys_1} using quadratic stability, which can be conservative. Less conservative stability methods consider slowly-varying parameter trajectories but assume asymptotic frozen-time stability \cite{ILCHMANN1987}. 

{
Next, we show that asymptotic frozen-time stability of the closed-loop system \eqref{math:closed_loop} can be achieved by incorporating the plant state into the controller dynamics.
Throughout this section, we consider the \emph{modified dynamic controller}, obtained from \eqref{math:dyn} by replacing $Y_K$ in \eqref{math:proj_grad} with $Y_{K,x}$ solving
\begin{align}
 \!  (A(\rho)\!-\!BK) Y_{K,x} \!+\! Y_{K,x} (A(\rho)\!-\!BK)^\top \!\!+ x x^\top \!\!+\! I_n\! =\! 0. \label{math:Ly_x}
\end{align}
This modification corresponds to the projected gradient flow of the cost   $ f_{K,x}(\rho) \!=\! \operatorname{tr}(P_K)+ x^{\!\top} \! P_K x $ with $\frac{\partial }{\partial K} f_{K,x}(\rho)\! =\! 2 (RK\!-\!B^\top \! P_K) Y_{K,x}$, which augments the state-independent cost $f_K$  with the initial-state dependent cost $f_K(x_0)$ \eqref{math:LQR_cost} evaluated at the current plant state $x_0\!=\!x$. Consequently, the controller dynamics \eqref{math:subsys_2} explicitly depend on the plant state.}

\begin{theorem}
   \label{theo:froze}
   Let Assumption \ref{ass:opt}, $ Q\succ 0$, and $\rho(t)=\rho\in\mathcal{P}$ hold.
   The equilibrium $\hat z = \begin{bmatrix} 0_n^\top , \operatorname{vec}^\top {(K^*_\rho )}\end{bmatrix}^\top$ of the closed-loop system \eqref{math:closed_loop} {with the modified dynamic controller} is asymptotically stable.
\end{theorem}

% \begin{theorem}
%    \label{theo:froze}
%    Let Assumption \ref{ass:opt}, $ Q\succ 0$, $\rho(t)=\rho\in\mathcal{P}$, and {$Y_K \leftarrow Y_{K,x}$} hold. 
%    The equilibrium $\hat z = \begin{bmatrix} 0_n^\top , \operatorname{vec}^\top {(K^*_\rho )}\end{bmatrix}^\top$ of the closed-loop system \eqref{math:closed_loop} is asymptotically stable.
% \end{theorem}

\begin{proof}
   {   Let $z=\begin{bmatrix}x^\top, \operatorname{vec}^\top (K )\end{bmatrix}^\top$ and   $V(z)= x^\top P_K x + \operatorname{tr}(P_K) - \operatorname{tr}(P_{K^*_\rho})$, where $P_K$ is as in \eqref{math:proj_grad}. 
   Then, $V(\hat z)=0$.
   % Since $A_K(\rho)$ is Hurwitz stable for all $ K \in  \mathcal{K}_{\mathcal{P}}$ and $K^\top R K + Q \succ 0$, $P_K \succ 0$ holds (see \cite[Th. 8.2]{hespanha2018}). 
   The term $\operatorname{tr }(P_K)-\operatorname{tr } {(P_{K_\rho^*})}= f_K(\rho)-f_{K^*_\rho}$ is positive for all $K\in\mathcal{K}_{\mathcal{P}}\setminus \{K^*_\rho \}$, because the LQR cost is strictly positive and has its minimum at $K^*_\rho$. 
   Therefore, $V(z)>0$ for all $z\in \mathbb{R}^n \times \mathcal{C}\setminus \{\hat z\}$. The time derivative of $V(z)$ along \eqref{math:closed_loop} is  
   \begin{align}
      \dot V(z) &= x^\top P_K \dot x + \dot{x}^\top P_K x  + x^\top \dot{P}_K x +  \operatorname{tr}(\dot P_K). \label{math:V_dot}
   \end{align}
   By using \eqref{math:Ly} with $A$ replaced by $A(\rho)$ and \eqref{math:subsys_1}, we obtain
   \begin{align}
      x^\top P_K \dot x + \dot{x}^\top P_K x  = - x^\top \left(Q + K^\top R K \right) x, \label{math:V_dot_1}
   \end{align}
   where $Q+K^\top R K \succ 0$. Moreover, 
   \begin{align}
    &x^\top \dot{P}_K x +  \operatorname{tr}(\dot P_K) =   \operatorname{tr} \bigl( \tfrac{\partial f_{K,x}(\rho)}{\partial K}^\top \dot{K} \bigr) \\
    &= -\alpha  \operatorname{vec}\bigl(\tfrac{\partial f_{K,x}(\rho)}{\partial K}\bigr)^\top M(K)  \operatorname{vec}\bigl( \tfrac{\partial f_{K,x}(\rho)}{\partial K} \bigr) \leq 0 \label{math:V_dot_2}
   \end{align} 
   because $M(K)\succ 0$ on $\operatorname{int}(\mathcal{C})$ by Lemma~\ref{lemma:5}, whose LICQ assumption is satisfied for $\mathcal{C}$ defined in \eqref{math:C}. Combining \eqref{math:V_dot_1} and \eqref{math:V_dot_2} yields $\dot{V}(\hat{z}) =0$, since $RK^*_\rho-B^\top P_{K^*_\rho}=0$, and $\dot V(z) < 0$ for all $z\in \mathbb{R}^n \times \operatorname{int}(\mathcal{C}) \setminus \{\hat z\}$,
   %  since $Q+K^\top R K\succ 0$
     and, for $x=0$, $Y_{K,0}=Y_K\succ 0$ and $RK\neq B^\top P_K$ for any $K\neq K^*_\rho$. Thus $V(z)$ is a Lyapunov function and $\hat{z}$ is asymptotically stable.}
\end{proof}

\begin{theorem}
   \label{theo:region}
   Let Assumption \ref{ass:opt}, $ Q\succ 0$, and $\rho(t)=\rho\in\mathcal{P}$ hold.
   The region of attraction (RoA) of the equilibrium $\hat z = \begin{bmatrix} 0_n^\top , \operatorname{vec}^\top \! {(K^*_{\rho} )}\end{bmatrix}^\top $  of the closed-loop system \eqref{math:closed_loop} {with the modified dynamic controller} is  $\mathbb{R}^n \times \operatorname{int}(\mathcal{C})$.
\end{theorem}

\begin{proof}
   {
Let $V(z)$ be the Lyapunov function from Theorem~\ref{theo:froze}. For $z(0)\in\mathbb{R}^n\times\operatorname{int}(\mathcal{C})$, set $c \coloneq V(z(0))$ and  $\Omega_c=\{z\in\mathbb{R}^n\times\mathcal{C}\mid V(z)\leq c\}$.
By Assumption~\ref{ass:opt}, $A_K(\rho)$ is Hurwitz on $\mathcal{C}$ and with $Q\succ0$ we have $P_K\succ0$. Hence, $p=\min_{\operatorname{vec}(K)\in\mathcal{C}}\lambda_{\min}(P_K)>0$, giving $\|x\|^2\leq c /p$ on $\Omega_c$, and since $\mathcal{C}$ is compact, $\Omega_c$ is compact.
From  $F - J_g J_g^\top = 2\operatorname{diag}(g)$ and
   $J_g M = 2 \operatorname{diag}(g) F^{-1} J_g$, we obtain
   $\dot{g}_i = \bigl(J_g \operatorname{vec}(\dot K)\bigr)_i=-2 \alpha\, g_i h_i$ with $h_i = \bigl[ F^{-1} J_g
   \operatorname{vec}\bigl(\tfrac{\partial f_{K,x}(\rho)}{\partial K}\bigr)\bigr]_i$
   along \eqref{math:closed_loop}.
Suppose $K(t)$ reaches $\partial\mathcal{C}$ at a first time $t_1>0$, i.e., $g_i(t_1)=0$ for some $i$. On $[0,t_1]$ we have $K(t)\in\mathcal{C}$, and $\dot{V}\leq0$ implies $V(z(t))\leq V(z(0))$, so $z(t)\in\Omega_c$. By compactness of $\Omega_c$ and continuity, $|h_i(z)|\leq L$ on $[0,t_1]$, hence $g_i(t_1)\geq g_i(0)e^{-2\alpha L t_1}>0$, contradicting $g_i(t_1)=0$. Thus $K(t)\in\operatorname{int}(\mathcal{C})$ and $z(t)\in\Omega_c$ for all $t\geq0$.
By LaSalle's invariance principle, $z(t)$ converges to the largest invariant
set $\mathcal{L}\subseteq \{ z \in \Omega_c \mid \dot{V}(z) = 0 \}$.
From~\eqref{math:V_dot_1}--\eqref{math:V_dot_2},
$\dot{V}(z) = 0$ requires both $x = 0$ (since $Q + K^\top R K \succ 0$) and
$M(K) \operatorname{vec}\bigl(\tfrac{\partial f_{K,x}(\rho)}{\partial K}\bigr) = 0$. Since $\operatorname{vec}(K(t)) \in \operatorname{int}(\mathcal{C})$, $M(K)\succ 0$ by Lemma~\ref{lemma:5}.
 For $x=0$, $Y_{K,0}=Y_K$ and $\frac{\partial f_{K,0}}{\partial K}=2(RK-B^\top P_K)Y_K$. With $Y_K\succ 0$ (as $A_K(\rho)$ is Hurwitz), $RK-B^\top P_K=0$ only holds for $K=K^*_\rho$. Hence $\mathcal{L}=\{\hat z\}$ and $z(t)\to\hat z$.
 Conversely, for $\operatorname{vec}(K(0))\in\partial\mathcal{C}$, i.e., $g_i(0)=0$
for some $i$, $g_i(t)=g_i(0)\exp\bigl(-2\alpha\int_0^t h_i \,\mathrm{d}s\bigr)=0$
for all $t\geq 0$, so $K(t)$ remains in $\partial\mathcal{C}$ and cannot converge
to $\operatorname{vec}(K^*_\rho)\in\operatorname{int}(\mathcal{C})$. Thus the RoA
is exactly $\mathbb{R}^n\times\operatorname{int}(\mathcal{C})$.}
\end{proof}

\begin{remark}
   Since the closed-loop system \eqref{math:closed_loop} is not an autonomous LPV system (see Remark~\ref{remark:not_LPV}), the stability methods for slowly-varying systems (see \cite{ILCHMANN1987}) cannot be applied directly.
   Under Assumption \ref{ass:opt}, the trajectory of the dynamic controller satisfies $ \operatorname{vec}(K(t) ) \in\mathcal{F}_{\mathcal{C}}^\mathcal{V}$ (see Proposition~\ref{prop:bound}) 
    and the feedback gain $K$ can be treated as a varying parameter. 
    Given the asymptotic frozen-time stability of the subsystem \eqref{math:subsys_1}, \eqref{math:subsys_1} is exponentially stable 
   if the rates of variation, i.e.,  $\Vert \dot{\rho}(t) \Vert $ and $\Vert \operatorname{vec}(\dot{K}(t) )  \Vert $,
   are sufficiently small (see \cite[Eq. (2.5)]{ILCHMANN1987}, \cite[Prop. 2.3.8]{briat_linear_2015} or \cite[Sec. 1.2.2.2]{mohammadpour2012control}). The rate $\Vert \operatorname{vec}(\dot{K}(t) )  \Vert $ decreases as the learning rate $\alpha$ in \eqref{math:dyn_1} is reduced.
\end{remark}

\section{Simulation Results}
\label{sec:4}
In this section, we demonstrate the stability of the closed-loop system and the transient performance of the dynamic controller.
Consider the following system 
\begin{align}
   \label{math:ex_system}
   \begin{split}
   A {(\rho(t))} &= \begin{bmatrix} -\rho(t) & 1 \\ -0.2 & 1\end{bmatrix},\; B = \begin{bmatrix} 1 \\ 0.5 \end{bmatrix}, \;Q = I_2,\; R = 2, 
   \end{split}
\end{align}
where  $K=\begin{bmatrix} K_1& K_2 \end{bmatrix}$, $\rho(t)\in\mathcal{F}^\infty_{\mathcal{P}}$ and $ \mathcal{P}=[0.5,2]$.
%The Assumption \ref{ass:standard} hold for the system \eqref{math:ex_system}.
By using Lemma~\ref{lemma:stable_poly}, the stability region $\mathcal{K}_\rho$ is characterized by the inequalities $f_1  = K_1 + \frac{K_2}{2}  + \rho -1 >0$ and $f_2  =\frac{K_2 \rho }{2}-\frac{K_2}{5}-\rho-\frac{K_1}{2}+\frac{1}{5}>0 $, and is visualized in Fig.~\ref{fig:1}.
\begin{figure}
   \centering
   \includegraphics[scale=0.8]{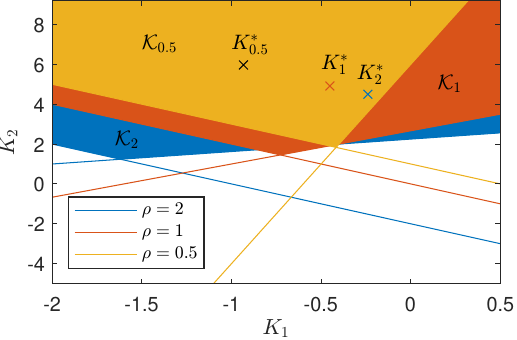}
   \caption{Stability regions $\mathcal{K}_\rho$, optimal feedbacks $K_\rho^*$ and stability boundaries $f_1=f_2 =0$ for $\rho=0.5, \rho=1$ and $\rho=2$.}
   \label{fig:1}
\end{figure}
The set $\mathcal{C}$ obtained by Proposition \ref{prop:1}, with $ \operatorname{vec}\left(\mathcal{K}^*  \right)\subseteq \operatorname{int}( \mathcal{C})$ (see Remark~\ref{rem:strict}), is given by $\mathcal{C}= [- 0.94, -0.23] \times [4.49, 5.97]$.
Moreover, the condition $\mathcal{C}\subseteq \operatorname{vec}\left(\mathcal{K}_{\mathcal{P}} \right) $ hold due to Proposition \ref{prop:2} and the obtained distances $d_i$ are depicted by solid black lines in Fig. \ref{fig:2}. 
\begin{figure}
   \centering
   \includegraphics[scale=0.8]{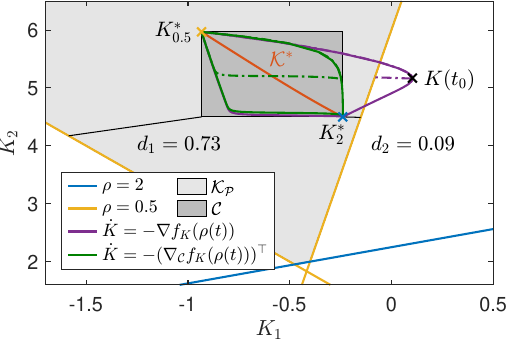}
   \caption{Visualization of the stability region $\mathcal{K}_{\mathcal{P}}$, the set $\mathcal{C}$ and trajectories of gradient flow and projected gradient flow for { $\rho=2 \to \rho = 0.5$}. }
   \label{fig:2}
\end{figure}
Quadratic stability of \eqref{math:ex_system} is verified by Theorem~\ref{theo:1} and the matrix $X =[3.1, -6.5 ; -6.5, 25.9]$ satisfies the LMI condition~\eqref{math:quad_2} for all $\tilde{A}_{K,i},i=1,\dots,8$. 

The trajectories of $\dot K = - \nabla f_K(\rho(t))$ and $\dot K = - (\nabla_{\mathcal{C}} f_K(\rho(t)))^\top$ are shown for $K(0)=K^*_{0.5}$ in Fig.~\ref{fig:2}, where the parameter switches from $\rho=2$ to $\rho=0.5$. For the solid lines, the switch occurs after convergence to $K^*_2$ and $K^*_{0.5}$ is reached.
For the dashed lines, the switch occurs prematurely at $t_0$, where the non-projected gradient flow yields $K(t_0)$ outside the stability region $\mathcal{K}_{0.5}$ (see Fig.~\ref{fig:2})%\footnote{This numeric example was specifically selected as a rare case where the non-projected gradient flow encounters ill-definedness.}
, leading to ill-defined behavior.
Strictly speaking, the integral of the LQR cost $f_K(\rho(t))$ diverges due to the unstable feedback $K(t_0)$, and the derivative $\nabla f_K(t_0)$ is not defined. However, using vectorization, a negative LQR cost is observed, which is impossible by definition, and the cost diverges to $-\infty$ at the boundary of $\mathcal{K}_{0.5}$, where the trajectory reaches a steady state. In contrast, the projected gradient flow remains well-defined and converges to $K^*_{0.5}$ after the switch at $t_0$.
\begin{figure}
   \centering
   \includegraphics[scale=0.7]{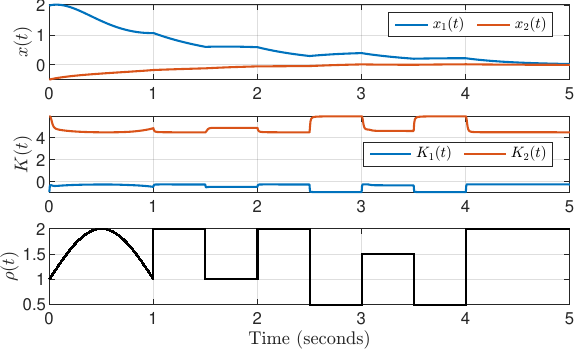}
   \caption{{Trajectories of the states $x(t)$ and the feedback gains $K(t)$ for the fast-varying parameter trajectory $\rho(t)$.}}
   \label{fig:3}
\end{figure}
% Figure \ref{fig:3} shows the trajectories of the states $x(t)$ and the feedback gains $K(t)$ for a fast-varying parameter trajectory $\rho(t)$. The fast-varying nature of $\rho(t)$ is characterized by steps where the derivative is infinite. The learning rate is set sufficiently high ($\alpha = 100$) to ensure that the feedback gains $K(t)$ quickly converge to the optimal feedbacks within each time interval of constant parameter trajectory. 

{
Fig.~\ref{fig:3} shows the trajectories of $x(t)$ and $K(t)$ for a fast-varying parameter trajectory $\rho(t)$. A high learning rate ($\alpha\!=\! 200$) ensures that $K(t)$ quickly converges to the optimal feedback during constant parameter intervals (see Remark~\ref{rem:5}). 
%Compared to a static feedback $K_{1.25}^*$, the LQR cost of the $\SI{6}{\second}$ simulation is reduced by $\SI{13.4}{\percent}$, demonstrating the effectiveness of the dynamic controller in enhancing transient performance. %It is observed that the transient advantage decreases as the switching frequency increases.
The dynamic controller yields an LQR cost of 12.46 over the $\SI{5}{\second}$ simulation, significantly lower than the 13.12 obtained with the static feedback controller $K_{\text{av}}=0.5 K_{0.5}^*+0.5 K_{2}^*$. This $\SI{5.1}{\percent}$ reduction demonstrates improved transient performance.}
\section{Conclusion}
\label{sec:conc}
This paper proposes a novel dynamic state-feedback controller for polytopic LPV systems based on a projected policy gradient flow that continuously minimizes the LQR cost while providing stability guarantees. We established conditions for quadratic stability of the LPV system that ensures exponential stability for arbitrarily fast-varying parameter trajectories. 
%Sufficient conditions were also provided to ensure bounded trajectories of the controller and convergence of these trajectories to the optimal feedback gains.
Sufficient conditions were also provided to ensure bounded trajectories of the controller and their convergence to the optimal feedback gains.
 { Additionally, we showed that augmenting the underlying cost of the controller with the LQR cost evaluated at the current plant state yields a modified dynamic controller that renders the frozen-time closed-loop systems asymptotically stable.}
 Simulation results showcased the effectiveness of the controller in maintaining stability and improving transient performance.
 %Future work will focus on tractable methods for over-approximating the set of optimal feedback gains.

\appendix

\subsection{Routh-Hurwitz Stability Criterion}
\label{subsec:routh}
%The Routh-Hurwitz criterion is a necessary and sufficient condition for the stability of LTI systems \cite[Sec. 6.2-6.5]{nise2019control}. 
The criterion states that a monic polynomial $p(s) = s^n + c_1 s^{n-1} + c_2 s^{n-2} + \ldots + c_{n-1} s + c_n$ has roots with negative real parts if and only if all the entries in the first column of the Routh table are strictly positive \cite[Sec. 6.2-6.5]{nise2019control}. %\footnote{We adopted the criterion to monic polynomials which appear in the characteristic polynomial.} 
The Routh table, with entries denoted by $R_{i,j}$, is constructed as follows%
\begin{align}
   \label{math:Routh}
   {%\small
      \begin{matrix} 
      s^n & 1 & c_2 & c_4 & c_6 & \cdots  \\ 
      s^{n-1} & c_1 & c_3 & c_5 & \cdots \\
      s^{n-2} & R_{3,1} & R_{3,2} & R_{3,3} & \cdots \\
     % s^{n-3} & R_{4,1} & R_{4,2} & R_{4,3} &  \\
      \vdots & \vdots & \vdots & \vdots \\
      s^1 & R_{n,1} & R_{n,2} & \\
      s^0 & c_n & 
   \end{matrix}},
\end{align}where $R_{i,j}=\left(R_{i-1,1} R_{i-2,j+1}-R_{i-2,1}R_{i-1,j+1}\right)/ R_{i-1,1}$ for $ 3 \!\leq \! i  \! \leq  \! n$, $1  \! \leq  \! j$,
%and missing terms, e.g., $R_{i-1,j+1}$, are treated as zero. 
and missing terms are treated as zero. 
%\footnote{For handling special cases where there is a zero in the first column or an entire row of zeros in the Routh table, refer to \cite[Section 6.3]{nise2019control}.}
%for $ 3 \leq i\leq n$.%, 1 \leq j \leq \lceil \tfrac{n-2 }{2}\rceil -i +2 $.
% The Routh table, with entries denoted by $R_{i,j}$, is formed by placing the coefficients of $p(s)$ with even powers of $s$ in the first row, i.e., $1,c_2,c_4,\dots$, and the coefficients with odd powers in the second row, i.e. $c_1,c_3,c_5,\dots$. The remaining entries are computed recursively as $R_{i,j}=\left(R_{i-1,1} R_{i-2,j+1}-R_{i-2,1}R_{i-1,j+1}\right)/ R_{i-1,1}$ for $ 3 \!\leq \! i  \! \leq  \! n$, $1  \! \leq  \! j$, where missing terms are treated as zero. 

\subsection{Proof of Lemma \ref{lemma:5}}
\label{app:lemma_5}

\begin{proof}
   Let ${\theta }\in\mathcal{M}$ and $w\in\mathbb{R}^{l}$ with $w^\top F({\theta })w = 2\sum_{i=1}^{l}g_i({\theta }) w_i^2 + \Vert J_g^\top({\theta }) w \Vert_2^2 = 0$. Both terms are nonnegative and hence vanish. { The first yields $w_i=0$ for all $i$ with $g_i({\theta })>0$, so that $0 = J_g^\top ({\theta }) w = \sum_{j\in I_0({\theta })} w_j \nabla g_j({\theta })$ with the active set $I_0({\theta }) = \{j \mid g_j({\theta }) = 0\}$. Since LICQ holds at ${\theta }$, $w_j = 0$ for all $j\in I_0({\theta })$. Hence $w=0$, proving $F({\theta })\succ 0$ on $\mathcal{M}$.}
   For the second claim, let $Z({\theta })\coloneq\left[\begin{smallmatrix} I_n & J_g^\top ({\theta }) \\ J_g({\theta }) & F({\theta })\end{smallmatrix}\right]$. Completing the square gives, for all $v\in\mathbb{R}^n$ and $w\in\mathbb{R}^l$,
   \begin{align}
      \begin{bmatrix}v \\ w \end{bmatrix}^{\!\top}\! \! Z({\theta })\! \begin{bmatrix}v \\ w \end{bmatrix} =\Vert v + J_g^\top({\theta }) w \Vert^2_2 + 2 w^\top \! \operatorname{diag }(g({\theta })) w. \label{math:z}
   \end{align}
   If ${\theta }\in \operatorname{int }(\mathcal{M})$, then $g_i({\theta })>0$ for all $i$, so \eqref{math:z} is positive for all $(v,w)\neq 0$ (for $w=0$ it reduces to $\Vert v \Vert_2^2$) and $Z({\theta })\succ 0$. If ${\theta }\in \partial \mathcal{M}$, then \eqref{math:z} is nonnegative and $Z({\theta })\succeq 0$.
    Since $F({\theta })\succ 0$, the Schur complement $M({\theta })$ of $F({\theta })$ in $Z({\theta })$ inherits the respective definiteness of $Z(\theta)$.
\end{proof}

\printbibliography

\end{document}